\documentclass[a4paper, 11pt]{article}

\usepackage{amsmath,amsfonts,amssymb,amsthm,graphicx,graphics,epsf}
\usepackage{graphicx}

\usepackage{microtype,thm-restate}%if unwanted, comment out or use option "draft"
\usepackage{algorithmic}
\usepackage{algorithm}
\usepackage{hyperref,color}
\usepackage[numbers]{natbib}
\usepackage{xspace}
\usepackage{lineno}
\graphicspath{{img/}} % No need to write this for every figure
\usepackage{geometry}

\newtheorem{theorem}{Theorem}[section]
\newtheorem{corollary}[theorem]{Corollary}
\newtheorem{lemma}[theorem]{Lemma}

\newtheorem{observation}[theorem]{Observation}

\title{Geodesic farthest-point Voronoi diagram in linear time}

\author{Luis Barba\thanks{Department of Computer Science, ETH Z\"urich, Switzerland, \texttt{luis.barba@inf.ethz.ch}}}
\date{} 

\newcommand{\icell}[1][i]{${#1}$-patch\xspace}
\newcommand{\icells}[1][i]{${#1}$-patches\xspace}
\newcommand{\idom}[1][i]{${#1}$-dominated\xspace}

\newcommand{\F}[2][P]{\ensuremath{F_{\scriptscriptstyle #1}(#2)}}
\newcommand{\s}{\mathcal S}
\newcommand{\g}[3][P]{\ensuremath{\textsc{g}^{\scriptscriptstyle #1}(#2, #3)}}
\newcommand{\dd}[3][P]{\ensuremath{\textsc{d}_{\scriptscriptstyle w}^{\scriptscriptstyle #1}(#2 {\leadsto} #3)}}
\newcommand{\ddw}[3][P]{\ensuremath{\textsc{d}_{\scriptscriptstyle w'}^{\scriptscriptstyle #1}(#2 {\leadsto} #3)}}
\newcommand{\p}[3][P]{\ensuremath{\pi_{_{#1}}(#2, #3)}}
\newcommand{\f}[2][P]{\ensuremath{f_{\scriptscriptstyle #1}(#2)}}
\newcommand{\funnel}[2][P]{\ensuremath{\mathtt{Funnel}_{\scriptscriptstyle #1}(#2)}}
\newcommand{\cell}[2][P]{\ensuremath{\mathtt{Cell}_{\scriptscriptstyle #1}(#2)}}

\newcommand{\bcell}[2][P]{\ensuremath{\mathtt{bCell}_{\scriptscriptstyle #1}(#2)}}
\newcommand{\interior}[1]{\mathrm{int}(#1)}

\newcommand{\vd}[2][P]{\textsc{vd}(#2, #1)}
\newcommand{\cost}[1]{\kappa(#1)}
\newcommand{\bvd}[2][P]{\textsc{vd}_{\partial}(#2, #1)}

\newcommand{\ex}[1]{\textsc{E}\left[#1\right]}
\newcommand{\exw}[1]{\textsc{E}#1}

\newcommand{\LL}[1][\s, P]{\ensuremath{\mathcal L_{_{#1}}}}
\newcommand{\A}{\ensuremath{\mathcal A}}

\begin{document}

\maketitle

\begin{abstract}
Let $P$ be a simple polygon with $n$ vertices.
For any two points in $P$, the geodesic distance between them is the length of the shortest path that connects them among all paths contained in $P$. 
Given a set $\s$ of $m$ sites being a subset of  the vertices of $P$, we present a randomized algorithm to compute the geodesic farthest-point Voronoi diagram of $\s$ in $P$ running in expected $O(n + m)$ time. 
That is, a partition of $P$ into cells, at most one cell per site, such that every point in a cell has the same farthest site with respect to the geodesic distance. 
In particular, this algorithm can be extended to run in expected $O(n + m\log m)$ time when $\s$ is an arbitrary set of $m$ sites contained in $P$, thereby solving the open problem posed by Mitchell in Chapter 27 of the Handbook of Computational Geometry.
\end{abstract}

\section{Introduction}
Let $P$ be a simple $n$-gon.
Let $\s$ be a set of $m\geq 3$ weighted \emph{sites} (points) contained in $V(P)$, where $V(P)$ denotes the set of vertices of $P$. 
That is, we have a function $w:\s\to \mathbb{R}$ that assigns to each site of $\s$ a non-negative weight. 
We also extend the weight function to any point in $P$ by setting $w(x) = 0$ for all $x\in P\setminus \s$.
While we could allow the sites to lie anywhere on the boundary, $\partial P$, of $P$, as long as we know their clockwise order along $\partial P$, we can split the edges of $P$ at the sites, and produce a new polygon where each site coincides with a vertex. Therefore, we assume that $\s\subseteq V(P)$.

Given two points $x,y$ in $P$ (either on the boundary or in the interior), the \emph{geodesic path} $\p{x}{y}$ is the shortest path contained in $P$ connecting $x$ with $y$. If the straight-line segment connecting $x$ with $y$ is contained in $P$, then $\p{x}{y}$ is the straight-line segment~$xy$. 
Otherwise, $\p{x}{y}$ is a polygonal chain whose vertices (other than its endpoints) are  reflex vertices of $P$. 
We refer the reader to~\cite{m-gspno-00} for more information on geodesic paths.

For a segment $xy$, we denote its Euclidean length by $|xy|$. For a path, its \emph{Euclidean length} is the sum of the Euclidean length of all of its segments. Given two points $x$ and $y$ in $P$, their \emph{geodesic distance} $\g{x}{y}$ is the Euclidean length of $\p{x}{y}$.
The \emph{weighted geodesic distance} (or simply \emph{$w$-distance}) between two points $x$ and $y$ in $P$, denoted by $\dd{x}{y}$, is the sum of $w(x)$ with the Euclidean length of $\p{x}{y}$, i.e., $\dd{x}{y} = w(x) + \g{x}{y}$. Notice that if all weights are set to zero, then the $w$-distance coincides with the classical definition of geodesic distance~\cite{m-gspno-00}. 
Moreover, notice that this distance is not symmetric unless the weights of $x$ and $y$ coincide. 

Given a point $x\in P$, an \emph{$\s$-farthest site} of $x$ in $P$ is a site  $s$ of  $\s$ whose $w$-distance to $x$ is maximized.
To ease the description, we assume that each vertex of $P$ has a unique $\s$-farthest neighbor. 
This \emph{general position} condition was also assumed in~\cite{ahn2015linear,aronov1993furthest,oh2016farthest} and can be obtained by applying a slight perturbation~\cite{edelsbrunner1990simulation}.

For a site $s\in \s$, let $\cell{s,\s} = \{x\in P: \dd{s}{x} \geq \dd{s'}{x}, \forall s'\in \s\}$ be the (weighted farthest) \emph{Voronoi cell} of $s$ (in $P$ with respect to $\s$). 
That is, $\cell{s,\s}$ consists of all the points of $P$ that have $s$ as one of their $\s$-farthest sites.
The union of all Voronoi cells covers the entire polygon $P$, and the closure of the set $\interior{P} \setminus \cup_{s \in \s} \interior{\cell{s,\s}}$ defines the (weighted farthest) \emph{Voronoi graph} of $\s$ in $P$. 

The Voronoi graph together with the set of Voronoi cells defines the \emph{weighted geodesic farthest-point
Voronoi diagram} (or simply \emph{FVD}) of $\s$ in $P$, denoted by $\vd{\s}$.
Thus, we indistinctively refer to $\vd{\s}$ as a graph or as a set of Voronoi cells; see Figure~\ref{fig:WeightsNotHarder}.

\begin{figure}[t]
\centering
\includegraphics{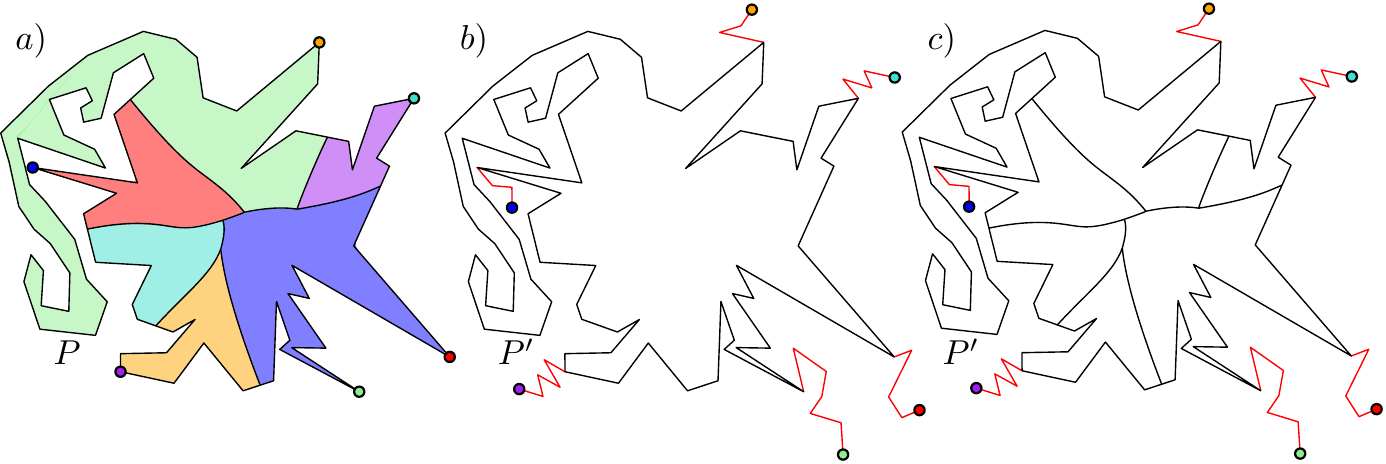}
%[width=1\textwidth]
\caption{$a)$ A simple polygon $P$ with a set $\s$ of six weighted sites and their FVD. $b)$ A new polygon $P'$ where a path of length $w(s)$ is added at the location of each site $s\in \s$. Then the weight is set to zero and moved at the endpoint of its corresponding path.
$c)$ The FVD of $\s$ in $P$ coincides with the FVD of the new sites in the new polygon $P'$.}
\label{fig:WeightsNotHarder}
\end{figure}

Notice that having non-zero weights on our set of sites does not make the problem harder. 
To see this, consider a new polygon $P'$, where at the location of each site $s\in \s$, a path of length $w(s)$ is attached to the boundary of $P$. 
Additionally, the site $s$ is given weight zero and is moved to the other endpoint of this path; see Figure~\ref{fig:WeightsNotHarder} for an illustration. 
In this way we obtain a new weakly simple polygon $P'$ and a new set of sites $\s'$ that defines the same FVD as $\s$ in $P$. 
Therefore, a weighted FVD as described in this paper has the same properties as the classical farthest-point Voronoi diagram constructed using the geodesic distance~\cite{aronov1993furthest}. 
In particular, we know that the Voronoi graph is a tree with leaves on the boundary of $P$.
Also, each edge of this graph consists of a sequence of straight-lines and hyperbolic arcs that may intersect $\partial P$ only at its endpoints~\cite{aronov1993furthest}. 
Thus, we refer to the Voronoi graph as a \emph{Voronoi tree} from now on.
While working with weighted sites might seem an unnecessary complication, we decided to work with them to ease the description of the recursive construction that our algorithm uses. 

Let $\F{x, \s}$ be the function that maps each $x\in P$ to the $w$-distance to a $\s$-farthest neighbor of $x$ (i.e., $\F{x, \s} = \dd{x}{\f{x, \s}}$).
Notice that $\F{x, \s}$ can be seen as the upper envelope of the $w$-distance functions from the sites in $\s$.
Throughout the paper, we will play with this alternative way of thinking of Voronoi diagrams as graphs or upper envelopes.
A point $c\in P$ that minimizes $\F{x, \s}$ is called the \emph{geodesic center} of $P$. 
Similarly, a point $s\in P$ that maximizes $\F{x, \s}$ (together with $\f{s}$) forms a \emph{diametral pair} and their $w$-distance is the \emph{geodesic diameter}.

\textbf{Related work.} 
The problem of computing the geodesic center of simple $n$-gon $P$ (and its counterpart, the geodesic diameter) were central in the 80's in the computational geometry community. 
Chazelle~\cite{c-tpca-82} provided the first $O(n^2)$-time algorithm to compute the geodesic diameter.
Suri~\cite{suri1989computing} improved upon it by reducing the running time to $O(n\log n)$. 
Finally, Hershberger and Suri~\cite{hershberger1993matrix} introduced a matrix search technique that allowed them to obtain a linear-time algorithm for computing the diameter.

The first algorithm for computing the geodesic center of $P$ was given by Asano and Toussaint~\cite{at-cgcsp-85}, and runs in $O(n^4\log n)$ time.  
This algorithm computes a super set of the vertices of the Voronoi tree of $\vd{\s}$, where $\s$ is the set of vertices of $P$.  
Shortly after, Pollack et al.~\cite{pollackComputingCenter} improved the running time to $O(n\log n)$.  
This remained the best running time for many years until recently when Ahn et al.~\cite{ahn2015linear} settled the complexity of this problem by presenting a $\Theta(n)$-time algorithm to compute the geodesic center of $P$. 

The problem of computing the FVD generalizes the problems of computing the geodesic center and the geodesic diameter.  
For a set $\s$ of $m\geq 3$ sites in a simple $n$-gon $P$, Aronov~\cite{aronov1993furthest} presented an algorithm to compute $\vd{\s}$ in $O((n+m)\log (n+ m))$ time.  
While the best known lower bound is $\Omega(n + m \log m)$, it was not known whether or not the dependence on $n$, the complexity of $P$, is linear in the running time.  
In fact, this problem was explicitly posed by Mitchell~\cite[Chapter 27]{m-gspno-00} in the Handbook of Computational Geometry, and solving it has become a prominent area of research in recent years.
Oh et al.~\cite{oh2016farthest} (SoCG'16) present the first improvement to this problem in more than 20 years. 
Using the new tools presented by Ahn et a.~\cite{ahn2015linear}, they introduce an $O(n \log\log n+m\log m)$-time algorithm to compute $\vd{\s}$. 
As a stepping stone, they present an $O((n+m)\log\log n)$-time algorithm for the simpler case where all sites are vertices of $P$.
In fact, any improvement on the latter algorithm translates directly to an improvement on the general problem.
In particular, a linear time algorithm for the simpler case with sites on the boundary of $P$ suffices to match the lower bound and close the problem presented by Mitchel~\cite[Chapter 27]{m-gspno-00}.

Recently, not only farthest-point Voronoi diagrams have received attention. 
For the nearest-point geodesic Voronoi diagram, two papers have focused in finding algorithms matching the same lower bound of $\Omega(n + m\log m)$~\cite{chihungVoronoi,oh2017voronoi}.
While the best results still work only for a limited range of $m$ with respect to $n$, both papers have appeared in consecutive years in the Symposium on Computational Geometry (SoCG).
However, the techniques in these papers use data structures that require logarithmic factors, and hence is not conceivable to transform them to obtain a linear algorithm for the case when the sites are vertices of the polygon. 

\textbf{Our results.} 
In this paper, we provide an optimal, albeit randomized, algorithm to compute $\vd{\s}$ for the special case where all sites of $\s$ are vertices of $P$.
Our algorithm runs in expected $\Theta(n + m)$ time. 
Using the reduction presented by Oh et al.~\cite{oh2016farthest}, we immediately obtain an algorithm for the general case where the sites can be arbitrary points in $P$.
This algorithm matches the lower bound and runs in expected $\Theta(n + m\log m)$ time thereby solving the problem posed by Mitchel~\cite[Chapter 27]{m-gspno-00}.
It remains open to find a deterministic algorithm with the same running.

\textbf{Our approach.}
Let  $P$ be a simple $n$-gon and let $\s$ be a set of $m\geq 3$ sites contained in $V(P)$, where $V(P)$ is the set of vertices of~$P$.
We present a randomized $O(n+m)$-time algorithm to compute the FVD of $\s$ in $P$.
We would like to use a variation of the randomized incremental construction (\emph{RIC}) for Euclidean farthest-point Voronoi diagrams~\cite{de2000computational}.
This algorithm inserts the sites, one by one, in random order and construct the cell of each newly inserted site in time proportional to its size. 
By bounding the expected size of each cell using backwards analysis, the incremental construction can be carried out in total linear time. 

In the geodesic case however, the complexity of a cell depends not only on the set of sites, but also on the complexity of the polygon~\cite{aronov1989geodesic}. 
Already the FVD of 3 sites can have $\Omega(n)$ vertices and arcs.
Moreover, there is an additional complication when using $w$-distances.
To achieve an incremental construction, one would need to have at hand a complete description of the $w$-distance function $\dd{s}{x}$ inside of the newly created cell for the inserted site $s$.
If this function is precomputed in the entire polygon, this would be too costly. 
Thus, one needs to define these functions only at the specific locations where they are needed. 
An additional problem is that for a RIC, the first inserted sites must have their $w$-distance defined in almost the entire polygon.
Thus, already the description-size of the $w$-distances needed for the first batch of sites (say the the first $m/100$) becomes super linear. 
Therefore, it seems hopeless to try a RIC without somehow reducing the complexity of $P$ throughout the process. 
Nevertheless, a RIC works great for all the sites that come after this first batch. 
Intuitively,  the latter insertions define smaller cells, and the space needed to describe their $w$-distances can be nicely bounded.
Thus, the main question is how to deal with this first fraction of the sites.

In this paper we overcome these difficulties with a novel approach, and manage to deal with this first fraction of the sites using pruning.
First, we partition randomly the sites into $B$ and $R$, where $|B| \leq \alpha m$ for some constant $0 < \alpha < 1$ (Section~\ref{section:First phase}).
Then, we construct recursively an ``approximation'' of the FVD of $B$ (Section~\ref{section:Smaller Polygon}). 
To this end, we define a new weakly simple polygon $Q$ containing $B$ with only a constant fraction of the vertices of $P$. 
Essentially we prune from $P$ all the vertices that have nothing to do with geodesic paths connecting sites in $B$ with points in their respective Voronoi cells. 
Our approximation comes from recursively computing the FVD of $B$ in $Q$. 
We show that the complexity of $Q$ decreases sufficiently so that the recursive call leads to a linear overall running time.

Reducing the complexity however comes with a price. The $w$-distance from sites  of $B$ inside of $Q$ turns to be only ``similar'' to that in $P$.
However, we make sure that these functions are accurate where it matters. 
After computing this ``Voronoi-like'' diagram for~$B$, we need to deal with the sites of $R$. 
To this end, we turn to the RIC (Section~\ref{section: Insertion process}).
We compute the $w$-distance from sites in $R$ only inside of specific parts of $P$, making sure that they suffice for our purpose, while their overall complexity remains linear.
Another challenge comes from the fact that the $w$-distances from $B$ are with respect to $Q$, while the ones from $R$ are not.
Thus, we need to prove that the upper envelope of these functions induces a Voronoi-like diagram.
Once we deal with these technical details, we end up with an upper envelope of functions that we prove to coincide with the FVD of $\s$ in $P$ finishing our construction. 

We show that the insertion of each site $r\in R$ can be carried out in expected $O(n/m)$ time.  
Thus, inserting all sites of $R$ can be done in expected $O(n+ m)$ time.
After inserting the sites of $R$, the expected total running time of our algorithm is given by the simple recurrence $\ex{T(n, m)} \leq T(n/2, m/2) + O(n+m) = O(n+ m)$.
The crucial aspect with our approach that could not be achieved before this paper, is the reduction in the complexity of the polygon. 
Overall, we combine many different tools, from recursion, pruning, and randomization, together with all the machinery to deal with geodesic functions. 
%Due to space constraints, all results marked with $[*]$ have their full proof in the Appendix. 

\section{Preliminaries}

Let $P$ be a simple $n$-gon and  let $\s$ be a set of $m\geq 3$ sites contained in $V(P)$. 
Because $\s\subset V(P)$, we know that $m = |\s|\leq  n$.

A subset $G\subseteq P$ is \emph{geodesically convex} in $P$ if for each $x,y\in G$, the geodesic path between $x$ and $y$ is contained in $G$, i.e., if $\p{x}{y}\subseteq P$.
Given a set $A$ of points in $P$, the \emph{geodesic hull} of $A$ in $P$ is the minimum geodesically convex set in $P$ that contains $A$. 
In particular, if $A\subseteq \partial P$, then the boundary of the geodesic hull of $A$ is obtained by joining consecutive points of $A$ along $\partial P$ by the geodesic path between them. Note that this geodesic hull is not necessarily a simple polygon but a weakly simple polygon. 
Geodesic functions in weakly simple polygons behave in the exact same way as in simple polygons, and the existent machinery applies directly with no overhead~\cite{chang2014detecting}. 
Thus, while many papers state their results for simple polygons, they applied directly to weakly simple polygons. 
In particular, all results and tools presented in this paper apply directly to weakly simple polygons.
This remark is already crucial in several recent papers~\cite{oh2016computing,oh2016farthest}.

\begin{lemma}[Restatement of Lemma 2 of~\cite{kpairpath}]\label{lemma:Geodesic hull computation}
Let $A\subseteq \partial P$ be a set of $O(n)$ points sorted along~$\partial P$. 
The geodesic hull of $A$ in $P$ can be computed in $O(n)$ time.
\end{lemma}

%An important consequence of the Lemma~\ref{lemma:Voronoi coincides for geodesically convex subsets} is the following result that will later allow us to simplify the structure of $P$ when computing FVDs. 
%
%\begin{corollary}\label{corollary:shortcut in Voronoi}
%For each $s\in \s$, if $a$ and $b$ are two points in $\partial P$ such that no site lies in $\partial P(a, b)$ and $a,b\in \cell{s, \s}$ for some $s\in \s$, then $\p{a}{b}\subseteq \cell{s, \s}$.
%\end{corollary}
%\begin{proof}
%Let $G$ be the geodesic hull of $A = \s\cup \{a, b\}$. 
%Assume for a contradiction that $\p{a}{b}$ intersects the interior of some other Voronoi cell. 
%Since $\vd[G]{\s}$ and $\vd{\s}$ coincide inside of $G$ by Lemma~\ref{lemma:Voronoi coincides for geodesically convex subsets}, in particular they coincide when restricted to the path $\p{a}{b}\subseteq G$. 
%Since $a$ and $b$ are two consecutive points of $A$ along $\partial P$, we know that $\p{a}{b}\subseteq \partial G$.
%Thus, the boundary of $\cell{s, \s}$ in $\vd[G]{\s}$ is disconnected---a contradiction with the fact that the FVD of any set of sites in a polygon is a connected tree. Therefore, we conclude that $\p{a}{b}\subseteq \cell{s, \s}$.
%\end{proof}

Let $\bvd{\s}$ be the FVD of $\s$ restricted to the boundary of $P$. 
More formally, for each $s\in \s$, let $\bcell{s, \s} = \cell{s, \s}\cap \partial P$ be the \emph{boundary cell} of $s$ and let $\bvd{\s}$ be the union of this boundary cells. The construction of $\bvd{\s}$ has always been a stepping stone in the computation of $\vd{\s}$~\cite{aronov1993furthest,oh2016farthest}, and in our algorithm it is not any different. 
The following result from~\cite{oh2016farthest} allows us to compute it efficiently.

\begin{theorem}[Theorem 9~\cite{oh2016farthest}]\label{thm:VD in boundary}
Let $P$ be an $n$-gon and let $\s\subseteq V(P)$ be a set of sites. 
Then, we can compute $\bvd{\s}$ in $O(n)$ time. 
\end{theorem}

Using the above procedure, we can find out in $O(n)$ time which sites of $\s$ have a non-empty Voronoi cell. 
Therefore, we can forget about the sites with empty cells and assume without loss of generality from now on that all sites of $\s$ have non-empty Voronoi cells.

Given a site $s\in \s$ and a polygonal chain $C \subseteq \partial P$ with endpoints $p$ and $p'$, the \emph{funnel} of $s$ to $C$ in $P$, denoted by $\funnel{s\to C}$, 
is the geodesic hull of $s$ and $C$ in $P$. 
It is known that $\funnel{s\to C}$ coincides with the weakly simple polygon contained in $P$ bounded by $C$, $\p{s}{p'}$ and $\p{s}{p}$~\cite{ahn2015linear}.
%Recall that the paths $\p{s}{p}$ and $\p{s}{p'}$ have a common subpath (maybe consisting only of $s$), and then they split never to meet again. 
%The \emph{tail} of $\funnel{s, \to C}$ is the path being the intersection of $\p{s}{p}$ and $\p{s}{p'}$. One endpoint of the tail is always $s$ and the other, called the \emph{anchor}, is the last point in which $\p{s}{p}$ and $\p{s}{p'}$ coincide when going from $s$ towards $C$. Note that this anchor is always a vertex of $P$.
For ease of notation, we denote $\funnel{s \to \bcell{s, \s}}$ simply by $\funnel{s, \s}$, i.e., the funnel with apex $s$ that goes to $\bcell{s, \s}$.
The following lemma shows the relation between Voronoi cells and their funnels.

\begin{lemma}[Consequence of Lemma 4.1 of~\cite{ahn2015linear}]\label{lemma:Voronoi cell in funnel}
Given a site $s\in \s$, the Voronoi cell $\cell{s, \s}$ is contained in the funnel $\funnel{s, \s}$.
\end{lemma}

We are also interested in bounding the total complexity of the funnels of sites in $\s$.
Given a polygon $Q$, let $|Q|$ denotes its \emph{combinatorial complexity} (or just \emph{complexity}), i.e., the number of vertices and edges used to represent it. 

\begin{lemma}[Consequence of Corollaries 3.8 and 4.4 of~\cite{ahn2015linear}]\label{lemma:Complexity of funnels}
Given an $n$-gon $P$ and a set $\s\subseteq V(P)$, $\sum_{s\in \s} |\funnel{s, \s}| = O(n)$.
Also, all funnels can be computed in $O(n)$ time. 
\end{lemma}

\subsection{The simplification transform}\label{section:Simplification}

\begin{figure}[t]
\centering
\includegraphics{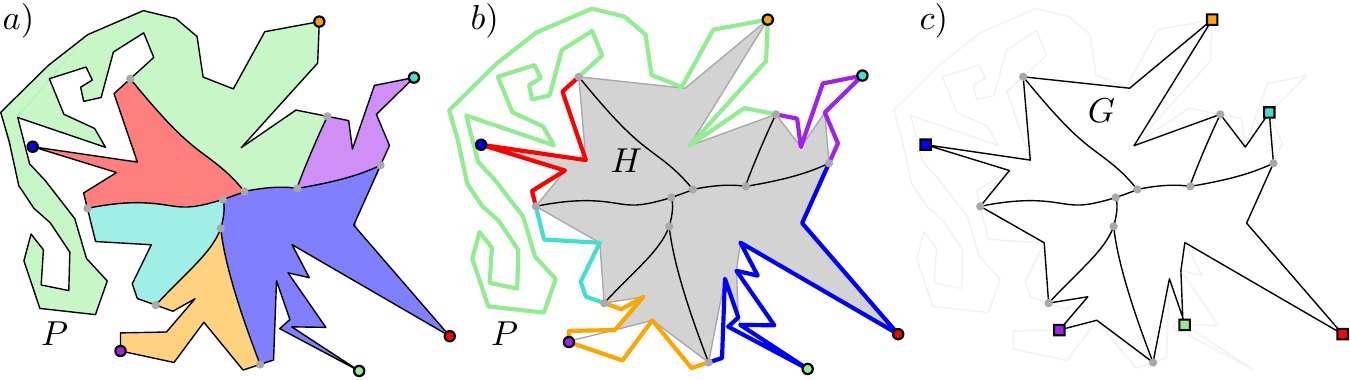}
%[width=1\textwidth]
\caption{$a)$ A simple polygon $P$ with a set $\s$ of six weighted sites and their FVD. $b)$ The polygon $H$ being the geodesic hull of $\LL$ and $\s$.
$c)$ The simplification transform allows to redefine the problem inside a simpler polygon $G$ with a new set of weighted sites and obtain the same FVD.}
\label{fig:FVD}
\end{figure}

The following transformation allows us to modify the input of our problem and assume some nice structural properties without loss of generality.
In this section and for this transformation, we allow the given polygon $P$ to be weakly simple instead of a simple $n$-gon. 
The result of the transformation described in this section takes a weakly simple polygon with a set of weighted sites as input, and produces a new simple polygon with a new set of weighted sites. Moreover, this resulting polygon has a particular structure that is crucial in the recursive calls of our algorithm.

Let $\LL$ be the set of leaves of $\vd{\s}$.
We first notice that we can focus on a specific geodesically convex subpolygon of $P$ to compute $\vd{\s}$.

\begin{lemma}\label{lemma:Voronoi coincides for geodesically convex subsets}
Let $H$ be the geodesic hull $\s\cup \LL$ in $P$.
Then, for each $s\in \s$, $\cell[H]{s, \s} \subseteq \cell{s, \s}$. Moreover, the Voronoi trees of $\vd[H]{\s}$ and $\vd{\s}$ coincide. 
\end{lemma}
\begin{proof}
Let $s$ be a site of $\s$ and let $x\in H$. 
Because $H$ is a geodesically convex subset of $P$, and since $x,s\in H$, we know that $\p[H]{x}{s} = \p{x}{s}$. 
That is, the $w$-distance to $x$ from each site in $\s$ is the same in $P$ and $H$. 
Therefore, the $\s$-farthest sites of $x$ are also preserved, which implies that $\cell[H]{s, \s} \subseteq \cell{s, \s}$. 
Because this happens for each Voronoi cell of $\s$ in $H$, and since the leaves belong also to $H$, the Voronoi trees coincide.
\end{proof}

While the geodesic hull $H$ of $\s\cup \LL$ in $P$ does not necessarily have lower complexity than $P$, it has some nice structure. 
We know that its boundary consists of geodesic paths that connect consecutive points in $\s\cup \LL$ along $\partial P$.
However, this geodesic hull $H$ is not necessarily a simple polygon; see Figure~\ref{fig:FVD}. 
To make it simple, we need to deal with ``dangling paths'' as follows.

We say that a vertex of $H$ is \emph{$H$-open} if it is incident to its interior.
For each $s\in \s$, let $a_s$ be the $H$-open vertex of $\funnel{s, \s}$ that is geodesically closest to $s$. 
Note that all paths from $s$ to any point in the interior of $H$ pass through $a_s$.
However, as long as the length of the geodesic path $\p{\s}{a_s}$ remains the same, the shape of this path is irrelevant. 
In fact, this is equivalent to giving $a_s$ a weight such that each distance measured from $a_s$ to points in the interior of $H$ has 
an added value of $\dd{s}{a_s}$. 

To formalize this intuition, we define a new polygon, a new set of sites, and new weighted distance function as follows.
Let $\A = \{a_s: s\in \s\}$ be the set of $m$ $H$-open vertices defined by $\s$. These vertices are our new set of sites.
Let $G(P, \s)$ (or simply $G$ if $P$ and $\s$ are clear from the context) be the geodesic hull of $\A\cup \LL$ in $P$.
Note that $G\subseteq H$ is a simple polygon by the definition of each $a_s$ in $\A$.
We define a new weight function $w':G\to \mathbb{R}$ so that $w'(x) = \begin{cases} \dd{s}{a_s}& \text{if }x=a_s\in \A \\ 0&\text{otherwise} \end{cases}$  (if $s$ and $a_s$ coincide, then their weights coincide).
With this new weight function, we can think of $\A$ as a set of weighted sites in $G$.
%Moreover, for each $x\in P$ and a site $a_s\in \A$, 
%we define a \emph{weighted distance function} $\dd{a_s}{x} = w(a_s) + \dd{a_s}{x} = \dd{s}{a_s} + \dd{a_s}{x} = \dd{s}{x}$.
%Because our algorithm is recursive, we can think of it dealing always with this extended weighted distance function.
%It is however for any practical matter a geodesic function.
%Thus, from now on we can assume that all our sites are weighted, and that distances are measured using the weighted distance function $\dd{s}{x}$ from any site $s\in \s$.

\begin{restatable}{lemma}{ApexFVD}\label{lemma:ApexFVD}
%$[*]$ 
It holds that $\vd[G]{\A}$ and $\vd[P]{\s}$ have the same Voronoi trees.
\end{restatable}
\begin{proof}
Notice that $G$ is a subset of $H$ such that $int(G)\subset int(H)$.
Let $s$ be a site of $\s$ and let $x\in G$. 
Because $G$ is a geodesically convex subset of $P$, and since $a_s, x\in G$, we know that $\p[G]{a_s}{x} = \p{a_s}{x}$, and hence the length of these paths is simply $\g{a_s}{x}$ . 
Therefore, we know that  $\ddw[G]{a_s}{x} = w'(a_s) + \g{a_s}{x} = \dd{s}{a_s} + \g{a_s}{x} = \dd{s}{x}$.
That is, the $w'$-distance from any site in $a_s \in \A$ to any point $x$ in $G$ is the same as the $w$-distance from the corresponding site in $s\in \s$ to the same point $x$ in $P$, which implies that $\cell[G]{s, \s} \subseteq \cell{s, \s}$. 
Because this happens for each Voronoi cell of $\s$ in $G$, the Voronoi trees coincide. 
\end{proof}
%}

By Lemma~\ref{lemma:ApexFVD}, we can always transform the problem of computing the FVD of $\s$ in $P$ as follows. 
Recall that $\LL$ is the set containing each leaf of $\vd{\s}$ and that $m= |\s|$.
Compute $\bvd{\s}$ and the funnel $\funnel{s, \s}$ of each site in $\s$ in total $O(n)$ time. 
Because $\s\cup \LL\subseteq \partial P$ and has size $2 m = O(n)$, we can compute $H$ the geodesic hull of $\s\cup \LL$ in $P$ in $O(n)$ time using Lemma~\ref{lemma:Geodesic hull computation}. 
After that, consider the set $\A$ of $H$-open vertices as defined above.
Again, we can compute the geodesic hull $G$ of $\A\cup \LL$ in $P$ in $O(n)$ time using Lemma~\ref{lemma:Geodesic hull computation}. 
By Lemma~\ref{lemma:ApexFVD}, $\vd[G]{\A}$ and $\vd[P]{B}$ coincide, so we can forget about $P$ and $S$, and focus simply on $G$ and $\A$ to compute the FVD.
We call this process the \emph{simplification transform}; see Figure~\ref{fig:FVD}. 
Note that the only convex vertices of $G$ are the sites in $\A$ and the leaves in $\LL$. 
We summarize the main properties of this simplification transform in the following result.

\begin{lemma}\label{lemma:Properties of simplification}
Let  $P$ be a simple $n$-gon and let $\s\subseteq V(P)$ be a set of $m\geq 3$ sites.
The simplification transform computes in $O(n)$ time a new simple polygon $G$ with at most $n+m$ vertices and a new set $\A\subseteq V(G)$ of $m$ weighted sites such that (1) the Voronoi trees of $\vd[G]{\A}$ and $\vd[P]{B}$ coincide, and (2) the set of convex vertices of $G$ is exactly $\A\cup \LL$.
\end{lemma}

\section{Computing the FVD}
Let  $P$ be a simple polygon and let let $\s$ be a set of $m\geq 3$ weighted sites contained in $V(P)$.
Using the simplification transform defined in Section~\ref{section:Simplification}, we can assume without loss of generality that $P$ is a simple polygon with at most $n+m$ vertices, and among them, its convex vertices are exactly the sites in $\s$ and the leaves of $\LL$ (see Lemma~\ref{lemma:Properties of simplification}). 
That is, it consists of at most $2m$ convex vertices.
If we consider consecutive vertices in $\s\cup \LL$ along $\partial P$, the chain connecting them consists only of reflex vertices of $P$, or is a single edge. 
The next step explained in the following section is to randomly partition $\s$.
Note that if $m = O(1)$, we can compute $\vd{\s}$ in $O(n)$ time by computing their bisectors and considering their overlay. Thus, from now on we assume that $m$ is larger than some predefined constant.

\subsection{First phase: the partition}\label{section:First phase}
We compute in linear time $\bvd{\s}$ using Lemma~\ref{thm:VD in boundary}, and let $\LL$ be the set of leaves of $\vd{\s}$. Note that $|\LL| = m$. 
For each $s\in \s$, we compute the funnel $\funnel{s, \s}$. 
Given a subset $R\subseteq \s$, let $\cost{R} = \sum_{r\in R} |\funnel{r, \s}|$ be the \emph{magnitude} of $R$. 
Lemma~\ref{lemma:Complexity of funnels} implies that $\cost{\s} = O(n)$, and that all these funnels can be computed in $O(n)$ time. 
To be more precise, let $\tau\geq 2$ be the constant hidden by the big $O$ notation, i.e., $\cost{\s} \leq \tau n$. 
Next, we compute a random permutation $\Pi$ of $\s$. 
Let $0 < \alpha < 1$ be some constant to be defined later. 
Let $B$ and $R$ be a partition of $\s$ such that $B$ consists of the first $\lfloor \frac{\alpha}{\tau} m\rfloor$ sites according to $\Pi$, and~$R = \s\setminus B$. 

\begin{observation}\label{obs: Complexity of B and R}
It holds that $\ex{\cost{B}} = \alpha n$ and $|B|= \lfloor \frac{\alpha}{\tau} m\rfloor \leq \alpha m$. 
\end{observation}

We would like to recursively compute a Voronoi-like diagram of the sites in $B$ while forgetting for a while of the red sites. 
Once we have this recursively computed diagram, we perform a randomized incremental construction of $\vd{\s}$ by inserting the sites of $R$ in the random order according to permutation $\Pi$.
In the next section we discuss the recursive call to compute a diagram for $B$, and later spend Section~\ref{section: Insertion process} detailing the insertion process.

\subsection{A smaller polygon}\label{section:Smaller Polygon}
While it would be great to compute $\vd{B}$, this may be too expensive as the diagram can have large complexity, and we need our recursive call to have smaller complexity (a constant fraction reduction in the size). 
Thus, we would not compute $\vd{B}$ exactly, but we will compute an ``approximation'' of it. 
Notice that we can see $\vd{B}$ as the upper envelope of the $w$-distances $\dd{b}{x}$. 
Because these functions have a complexity that depends on the size of the polygon, we need to simplify them.
To achieve this, for a site $b\in B$, this simpler distance function will be completely accurate inside of $\cell{b, \s}$. 
However, for any point $x$ outside of $\cell{b, \s}$, this new distance from $s$ to $x$ will be only upper bounded by $\dd{b}{x}$. 
That is, distances from $b$ can only get shorter, and only outside of $\cell{b, \s}$. 

To define these new distance functions, we define a new polygon $Q$ of lower complexity than $P$ (although $P\subseteq Q$).
Let $V_B$ be the set consisting of all vertices of $P$ that belong to the funnel $\funnel{b, \s}$ of some $b\in B$.
By Observation~\ref{obs: Complexity of B and R}, we know that the expected size of $V_B$ is at most $\alpha n$. 
Note that we could repeat the construction of $B$ and $R$ an expected constant number times, until we guarantee that $V_B\leq \alpha n)$.
Let $V_C$ be the set of convex vertices of~$P$. Recall that by our assumption that the simplification transform has already been applied to~$P$, 
we know by Lemma~\ref{section:Simplification} that $V_C$ consists of the union of $\s$ and $\LL$, where $\LL$ is the set of leaves of $\vd{\s}$, i.e., $|V_C| \leq 2m$.
Let $Q$ be a polygon defined as follows. 
Imagine the boundary of $P$ being a rubber band, and each vertex of $V_B\cup V_C$ being a pin.
By letting the rubber band free while keeping it attached at the pins, this rubber band snaps to a closed curve defining a weakly simple polygon $Q$; see Figure~\ref{fig:Reduction}.

\begin{figure}[t]
\centering
\includegraphics{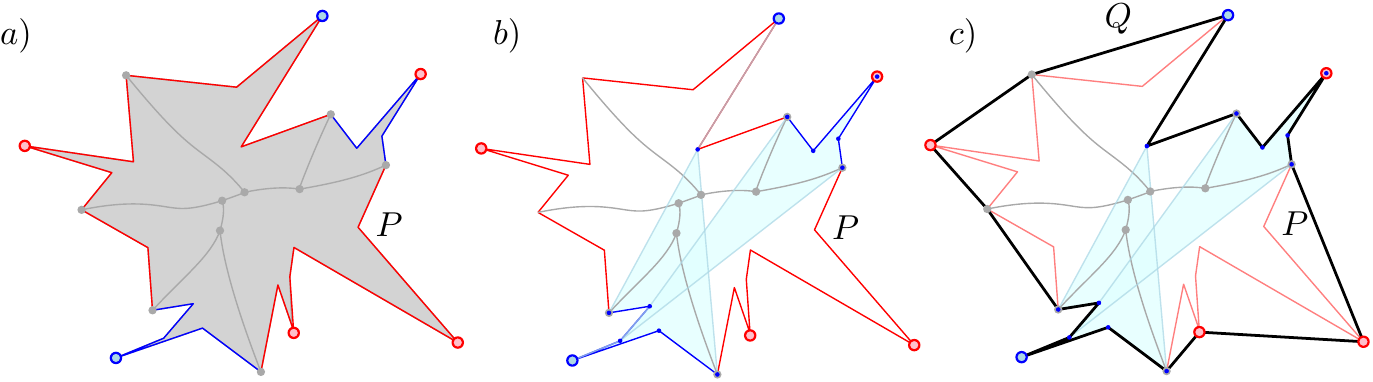}
%[width=1\textwidth]
\caption{$a)$ The polygon obtained from the simplification transform, and the decomposition of its boundary into red and blue chains.
$b)$ The funnels of the sites in $B$ are depicted, as well as all vertices in $V_B$.
$c)$ The polygon $Q$ is obtained by taking a rubber band and keeping attached at all vertices in $V_B\cup V_C$ and letting it snap. }
\label{fig:Reduction}
\end{figure}

\begin{restatable}{lemma}{PropertiesOfQ}\label{lemma:PropertiesOfQ}
%$[*]$ 
The polygon $Q$ contains $P$, is weakly simple, can be computed in expected $O(n)$ time, and has at most $\alpha n + 4m$ vertices.
Moreover, for $x,y\in P$, it holds that $\g[Q]{x}{y} \leq \g{x}{y}$. 
In particular, for each $b\in B$ and $x\in P$, $\dd[Q]{b}{x} \leq \dd{b}{x}$, and if $x\in \cell{b, \s}$, then $\dd[Q]{b}{x} = \dd{b}{x}$.
\end{restatable}
%\begin{proof}[Sketch proof]
%The boundary of $Q$ can be constructed by connecting consecutive  points in $V_B$ by geodesics contained in the complement of $P$, i.e., in a domain that has $P$ as a hole or obstacle ($Q$ is also known as the relative hull of $V_B$ in this domain). Only convex vertices of $P$ can be in these paths, and they can be visited only twice.
%Therefore, $Q$ consists of at most $|V_B| + 2m = \alpha n + 2m$ reflex vertices; see full proof in Appendix. 
%\end{proof}
%\newcommand{\ProofPropertiesOfQ}{
%\PropertiesOfQ*
\begin{proof}
The boundary of $Q$ can be constructed by connecting consecutive  points in $V_B\cup V_C$ by geodesics contained in the complement of $P$, i.e., in a domain that has $P$ as a hole or obstacle ($Q$ is also known as the relative hull of $V_B\cup V_C$ in this domain). 
Note however that we do not need to deal with geodesics in domains with holes. 
We can take a box sufficiently large to enclose $P$, and connect it using a small corridor with the left most vertex of $P$ to create a new polygon $P^\circ$ that essentially is a box with $P$ as a ``hole''; see Figure~\ref{fig:ConstructingQ}. 
Using this polygon that can be constructed in linear time, we can compute the boundary of $Q$ in $O(n)$ time~\cite[Lemma 2]{kpairpath}.

Note that the reflex vertices in $P$ become convex in $P^\circ$ and vice versa. 
Therefore, the geodesic paths in $P^\circ$ connecting consecutive vertices of $V_B\cup V_C$ consist only of reflex vertices of $P^\circ$ other than their endpoints, i.e., they consist only of vertices of $V_C$.  
Moreover, each vertex of $V_C$ can be visited only twice by the boundary of $Q$, and hence, it creates at most one new reflex vertex of $Q$.
Therefore, $Q$ consists of at most $|V_B| +  2|V_C| \leq \alpha n + 4m$ vertices.

For $x,y\in P$, the geodesic $\p{x}{y}$ is also contained in $Q$. Thus, their geodesic distance in $Q$ can only get shorter, i.e., $\g[Q]{x}{y} \leq \g{x}{y}$.
Moreover, for some $b\in B$, notice that $V_B$ contains all vertices of $\funnel{b, \s}$. Therefore, for $x\in \cell{b,\s}$, all vertices in the geodesic $\p{b}{x}$ are still on $Q$, i.e., $\dd[Q]{b}{x} = \dd{b}{x}$. Moreover, if $x\in P$, but not in $\cell{b,\s}$, then the path to $x$ can only become shorter, i.e., $\dd[Q]{b}{x} \leq \dd{b}{x}$.
\end{proof}
%}

\begin{figure}[t]
\centering
\includegraphics{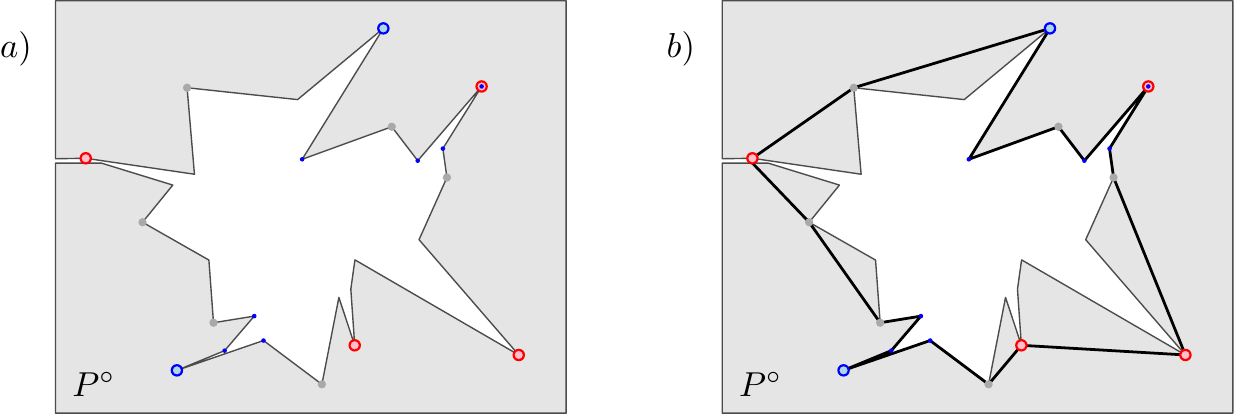}
%[width=1\textwidth]
\caption{The algorithmic construction of $Q$ described in the proof of Lemma~\ref{lemma:PropertiesOfQ}.}
\label{fig:ConstructingQ}
\end{figure}

Our plan is now to compute the FVD of $B$ in the new polygon $Q$, and then use this as a ``good'' approximation of $\vd{B}$ in $P$. 
With good, we mean that the red sites can be randomly inserted in this diagram, and that the result of this whole process is indeed $\vd{\s}$. 
We will prove this properties in the next section, but for now, we focus on describing the recursive algorithm.

%By Lemma~\ref{lemma:PropertiesOfQ}, $Q$ has at most $\alpha n + 2m$ reflex vertices.
%Since $m$ could be large (up to $n$), we need to further refine $Q$ before being able to apply recursion. 
%To this end, we just apply the simplification transform on $Q$ and the sites in~$B$. 
%In this way, we obtain a new polygon $G$ and a new set of sites $\A$ such that the Voronoi trees of $\vd[Q]{B}$ and $\vd[G]{\A}$ coincide.
%Recall that to compute $G$, we compute $\bvd[Q]{B}$ and obtain the set $\LL[B,Q]$ of leaves of $\vd[Q]{B}$. 
%Then, we look at the clockwise order of the points in $\LL[B,Q]\cup \A$ along $\partial P$, and compute all geodesic between consecutive points according to this order. 
%The union of these geodesic paths forms the boundary of $G$. 
%Thus, the convex vertices in $G$ are exactly the points in $\LL[B,Q]\cup \A$. 
%Since $|\LL[B,Q]| = |B| = |\A| \leq \alpha m$, we conclude that $G$ consists of at most $\alpha n + 2m$ reflex vertices (the reflex vertices in $Q$) and at most $|\A| + |\LL[B,Q]|\leq 2\alpha m$ convex vertices. 
Let $I(n,m)$ be the time to insert back the sites of $R$ and obtain $\vd{\s}$ after having recursively computed $\vd[Q]{B}$.
Because $Q$ consists of at most $\alpha n + 4m$ vertices by Lemma~\ref{lemma:PropertiesOfQ}, and since $|B|\leq \alpha m$ by Lemma~\ref{obs: Complexity of B and R}, we get a recursion of the form $T(n, m) = T(\alpha n + 4m, \alpha m) + I(n,m)$ for the running time of our algorithm. 
We claim that $I(n,m) = O(n+m)$, and we prove it in the next couple of sections.
However, for $T(n,m)$ to solve to $O(n+ m)$, we need to look at one more iteration of the recursion, as it can be that $\alpha n + 4m$ is not really smaller than $n$ if $m$ is large. 
Fortunately, because $T(\alpha n + 4m, \alpha m) = T(\alpha(\alpha n + 4m) + 4\alpha m, \alpha^2 m) + I(\alpha n + 4m, \alpha m)$, and by our assumption on the running time of $I(n,m)$, we get that
\[T(n,m) = T(\alpha^2 n + 8\alpha m, \alpha^2 m) + O(n+m) + O(\alpha n + 4m + \alpha m).\]

By choosing the constant $\alpha$ sufficiently small, and since we assume that $m\leq n$, we can guarantee that $T(n, m) \leq T(n/2, m/2) + O(n + m) = O(n +m)$ proving the main result of this paper.
Therefore, it remains only to show that $I(n,m)$ is indeed $O(n+ m)$, i.e., in linear time we can insert back the red sites, and obtain the FVD $\vd{\s}$ from the recursively computed diagram of $\vd[Q]{B}$.

\subsection{Preprocessing the red sites}\label{sec: Preprocessing of red sites}
Before going into the insertion process of the sites in $R$, we need to finish some preprocessing on them.
To be able to insert these sites efficiently, we need to have a representation of the $w$-distance of each $r\in R$ defined on a sufficiently large superset of $\cell{r, \s}$. 

Note that we cannot define these distance functions in the entire polygon, otherwise we are spending already too much time and space.
On the other hand, if its representation is too narrow, then during the insertion it might be that the distance information is insufficient.

Recall that $\LL$ denotes the set of leaves of $\vd{\s}$.
Color the leaves in $\LL$ \emph{purple} if they bound the Voronoi cell of a site in $B$ and a site in $R$; see Figure~\ref{fig:RedPreprocessing}.

\begin{figure}[ht]
\centering
\includegraphics{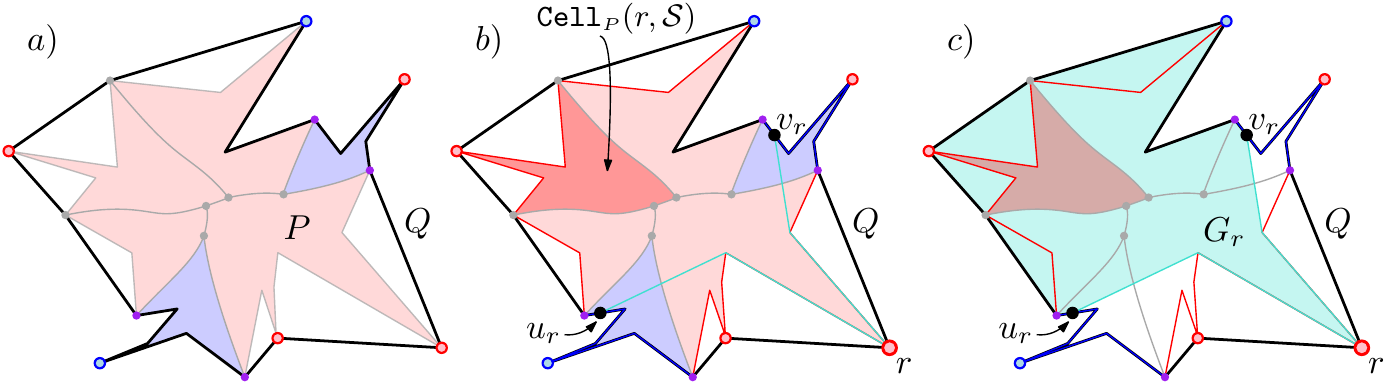}
%[width=1\textwidth]
\caption{$a)$ The coloring of the purple leaves of $\LL$.
$b)$ For a site $r\in R$, the construction of $u_r$ and $v_r$ lying inside blue Voronoi cells.
$c)$ The polygon $G_r$ where we compute the SPM of $r$.}
\label{fig:RedPreprocessing}
\end{figure}

%We then partition $\partial P$ by splitting at each purple point (the purple point belongs to both sides of the split). 
%In this way, we obtain a collection of maximal polygonal chains without interior purple points called \emph{$B$-$R$-chains}.
%We group the red sites into blocks as follows. A site $r\in R$ belongs to the \emph{block} $R(C)$ of a $B$-$R$-chain $C$ if $\bcell{r, \s}$ is contained in $C$.

Recall that given a polygon $K$ and two points $x$ and $y$ on $\partial K$, $\partial K(x,y)$ denotes be the polygonal chain that starts at $x$ and follows the boundary of $K$ clockwise until reaching~$y$.
For each $r\in R$, let $u_r$ and $v_r$ respectively be the first purple leaves of $\LL$ reached from any point in $\bcell{r, \s}$ when walking counterclockwise and clockwise along $\partial P$. 
We then move $u_r$ and $v_r$ slightly clockwise and counterclockwise, respectively, so that they both sit inside of a blue Voronoi cell. 
Moreover, we know that $\bcell{r,\s}$ is contained in the interior of the path $\partial P(u_r,v_r)$. 
We define a new polygon where we define the $w$-function of $r$ as follows.
Because both $u_r$ and $v_r$ lie on $\partial P$ and on Voronoi cells of blues sites, we know that $u_r$ and $v_r$ are both on $\partial Q$. 
Let $G_r$ be a new weakly simple polygon bounded by the paths $\p{r}{u_r}, \partial Q(u_r, v_r)$ and $\p{v_r}{r}$. 
Notice however that the paths $\p{r}{u_r}$ and $\p{v_r}{r}$, called the \emph{walls}, are defined within the polygon $P$, while the other path bounding $G_r$ is contained in $\partial Q$. 
Since $P\subset Q$, we know that $G_r\subseteq Q$. 
Moreover, the funnel $\funnel{r\to \partial P(u_r, v_r)}$ is contained in $G_r$; see Figure~\ref{fig:RedPreprocessing}.

\begin{restatable}{lemma}{SPMForSitesR}\label{lemma:SPMForSitesR}
%$[*]$ 
Given $r\in R$, it holds that $\cell{r, \s}\subseteq G_r$. 
Moreover, it holds that 
\[ \ex{\sum_{r\in R} |G_r|} = O(n)\text{ and }\bigcup_{r\in R} G_r = Q.\]
\end{restatable}
%\begin{proof}[Proof sketch]
%\end{proof}

%\newcommand{\ProofSPMForSitesR}{
%\SPMForSitesR*  
\begin{proof}
Recall that by Lemma~\ref{lemma:Voronoi cell in funnel}, we know that $\cell{r, \s}$ is contained in the funnel $\funnel{r, \s}$.
Moreover, $\funnel{r, \s}$ is contained in $\funnel{r\to  \partial P(u_r, v_r)}$, which in term is contained in $G_r$ by construction.
Therefore, $\cell{r, \s}\subseteq G_r$ as claimed. 

%Notice that $\bcell{r, \s} = \bcell{r, B\cup R(C)}$ as all the sites limiting the boundary of $\bcell{r,\s}$ are in $B\cup R(C)$. 
%Thus, Lemma~\ref{lemma:Voronoi cell in funnel} implies that $\cell{r, B\cup R(C)} \subseteq \funnel{r, B\cup R(C)} = \funnel{r \to \bcell{r, B\cup R(C)}}$.
%Moreover, because $C$ contains $\bcell{r, \s} = \bcell{r, B\cup R(C)}$, and since $G_C$ is geodesically convex and contains $r$, 
%we conclude that $\cell{r, B\cup R(C)} \subseteq \funnel{r, B\cup R(C)} \subseteq G_C$ as claimed.

We use some concepts and results introduced by Ahn et al.~\cite[Section 3]{ahn2015linear}.
A \emph{transition chain} is a polygonal chain contained in $\partial P$ such that its endpoints belong to different Voronoi cells. 
Because  $u_r$ and $v_r$ lie in the Voronoi cells of different blues sites, 
we get that $\partial P(u_r, v_r)$ is a transition chain. 
Let $R_r$ be the set of sites in $\s$ whose Voronoi cell intersects~$\partial P(u_r, v_r)$. 
In particular, $R_r$ contains $r$, two blue sites, and it might contain other sites of $R$.
Let $H_r$ be the geodesic hull of $\partial P(u_r, v_r) \cup R_r$. 
Note that $\funnel{r\to \partial P(u_r, v_r)} \subseteq H_r$, and that both paths $\p{r}{u_r}$ and $\p{r}{v_r}$ are shared by the boundaries of $\funnel{r\to \partial P(u_r, v_r)}$ and $G_r$.
We call these paths the \emph{walls} of these funnels. 
Recall that the third path of $G_r$ is $\partial Q(u_r, v_r)$, and notice that $\cup_{r\in R} \partial Q(u_r, v_r) = \partial Q$. 
Because $|Q| = O(n)$, to upper bound the total complexity of all the polygons $G_r$, 
it remains only to upper bound the complexity of their walls. 
That is, it remains to upper bound the complexity of all polygons $H_r$. 
To this end, let $\mathcal H = \{H_r : r\in R\}$. Note that since this is a set, if the same polygon is the same for many sites in $R$, it is counted only once in $\mathcal H$.
Because each $\partial P(u_r, v_r)$ is a transition chain, a result from Ahn et al.~\cite[Lemma 3.6]{ahn2015linear} implies directly that $\sum_{H\in \mathcal H} |H| = O(n)$ (the hourglass defined on the transition chain $\partial P(u_r, v_r)$ is a superset of $H_r$). 
Therefore, the only thing that remains is to show that each polygon $H\in \mathcal H$ appears with $O(1)$ multiplicity in expectation, or equivalently, that the expected size of $R_r$ is constant.

To show that the expected size of $R_r$ is $O(1)$, fix a site $r\in R$. 
For each $s\in \s$, let $\chi^r_s$ be an indicator random variable (\emph{i.r.v.}) that is one if and only if $s$ is belongs to~$R_r$.
Note that for $\chi^r_s$ to be one, all sites between $r$ and $s$ along $\partial P$ must be in $R$.
For a site $s\in \s$, let $rank(s)$ denote its rank in the permutation $\Pi$. 
Note that the probability that $s$ appears after $r$ according to $\Pi$ is $\frac{m-rank(r)}{m}$.
Thus, $Pr[\chi^r_s = 1] \leq 2\left(\frac{m-rank(r)}{m}\right )^{\delta(r,s)}$, where $\delta(r,s)$ is the number of sites of $\s$ visited when going clockwise from $r$ to $s$ along $\partial P$. 
The two factor comes from walking clockwise and counterclockwise from $r$ along $\partial P$, both directions are symmetric.
Note however that since $r\in R$, $rank(r) > \frac{\alpha}{\tau} m$ by Observation~\ref{obs: Complexity of B and R}.
Thus, $Pr[\chi^r_s = 1] \leq 2\left(\frac{\tau - \alpha}{\tau}\right)^{\delta(r,s)}$.
By definition, we know that $|R_r| = \sum_{s\in \s} \chi^r_s$.
Applying expectations, we get that $\exw{|R_r|} \leq 2 \sum_{s\in \s} \left(\frac{\tau - \alpha}{\tau}\right)^{-\delta(r,s)}$. 
Note that we can order the sites of $\s$ according to their distance $\delta(r, s)$.
Thus, we can rewrite the summation as \linebreak
$\exw{ |R_r| } \leq 2 \sum_{i = 0}^{m-1} \left(\frac{\tau - \alpha}{\tau}\right)^{-i}  = O(1).$
Consequently, we conclude that $ \ex{\sum_{r\in R} |G_r|} = O(n)$.

It remains only to show that $\bigcup_{r\in R} G_r = Q$. To this end, note that $\cell{r, \s}$ is contained in $G_r$, and hence, the union of all the $G_r$'s covers all the Voronoi cells in $\vd{\s}$.
That is, $P$ is contained in $\bigcup_{r\in R} G_r$. We show now that $Q\setminus P$ is also contained in it. 
Recall that to define $G_r$, we use $\partial Q(u_r, v_r)$ instead of $\partial P(u_r, v_r)$ to close the boundary of $G_r$. The space between $\partial Q(u_r, v_r)$ and $\partial P(u_r, v_r)$ encloses a region of $Q\setminus P$, say $W_r$. Moreover, if we take the union of all these regions, then we get that $Q\setminus P = \cup_{r\in R} W_r$ proving our result.
\end{proof}

The last technical detail is the structure used to store the $w$-distances from the sites in $R$.
Let $s\in \s$ and let $H\subseteq P$ be a subpolygon such that $s\in H$.
The \emph{shortest-path map} (or \emph{SPM} for short) of $s$ in $H$ is a subdivision of $H$ into triangles such that the geodesic path to all the points in one triangle has the same combinatorial structure. 
By precomputing the geodesic distance to each vertex of $H$, we get a constant-sized representation of the $w$-distance from $s$ inside each triangle (for more information on shortest-path maps refer to~\cite{guibasShortestPathTree}). We know also that the SPM of $s$ in $H$ can be computed in $O(|H|)$ time~\cite{chazelle1991triangulating,guibasShortestPathTree}. 

Using this SPM's, we describe our $w$-distances as follows.
For each $r\in R$, we compute the SPM of $r$ in $G_r$.
Because the complexity of this SPM is $O(|G_r|)$, we conclude that the total expected complexity of all these SPM's is $\ex{\sum_{r\in R} |G_r|} = O(n)$ by Lemma~\ref{lemma:SPMForSitesR}.

\section{Inserting back the red sites}\label{section: Insertion process}

After computing $\vd[Q]{B}$ recursively, we would like to start the randomized incremental construction of sites of $R$. 
But first, we should specify how do we store $\vd[Q]{B}$. 
%Let $s\in \s$ and let $H\subseteq P$ be a subpolygon such that $s\in H$.
%The \emph{shortest-path map} (or \emph{SPM} for short) of $s$ in $H$ is a subdivision of $H$ into triangles such that the geodesic path to all the points in one triangle has the same combinatorial structure. 
%By precomputing the geodesic distance to each vertex of $H$, we get a constant-sized representation of the geodesic distance function from $s$ inside each triangle (for more information on shortest-path maps refer to~\cite{guibasShortestPathTree}). We know also that the SPM of $s$ in $H$ can be computed in $O(|H|)$ time~\cite{chazelle1991triangulating,guibasShortestPathTree}. 

%Let $\triangle$ be an apexed triangle in this partition and let $a$ be its \emph{apex}, i.e., the vertex of $\triangle$ with smaller geodesic distance to $s$. 
%It holds then that for each $x\in \triangle$, $\p{s}{x} = \p{s}{a}\cup \p{a}{x}$. 
%In addition, the geodesic distance from $s$ to the apex of each apexed triangle is precomputed. 
%Thus, after locating a point $x$ inside an apexed triangle  $\triangle$, we can compute $\dd{s}{x} = \dd{s}{a} + |ax|$ in $O(1)$ time  
%We know also that the SPM of $s$ in $H$ can be computed in $O(|H|)$ time~\cite{chazelle1991triangulating,guibasShortestPathTree}.

Using the SPM's, we introduce the \emph{refined FVD} of $B$ in $Q$.
This refined FVD is a decomposition of $Q$ into constant-size cells defined as follows:
For each site $b\in B$, the Voronoi cell $\cell[Q]{b, B}$ is subdivided by the defining triangles of the SPM of $b$ in $Q$. 
That is, we take the intersection of each defining triangle $\triangle$ of the SPM of $b$ and intersect it with the Voronoi cell of $b$ to obtain a \emph{refined triangle}. 
As usual, for each point in a refined triangle, the geodesic distance is measured from its apex $a$ and added with $\dd{b}{a}$. 
Thus, each refined triangle and its distance function can be described with $O(1)$ space. We say that $b$ \emph{owns} these refined triangles. 
In other words, we have a way to describe the upper envelope of the $w$-distances of the sites in $B$ within $Q$ using a collection of constant-size refined triangles. 
%From the linear time construction of the SPM of each $b\in B$ within $Q$, and from the FVD of $B$ in $Q$, we can obtain its refined FVD in additional $O(|Q|)$ time.

Assume inductively that $\vd[Q]{B}$ is represented as a refined FVD as described above.
That is, for each site $b$ of $B$, there is a collection of refined triangles owned by $b$ which cover the entire Voronoi cell $\cell[Q]{b, B}$. 
Let $f_b:Q \to \mathbb{R}$ such that $f_b(x) = \dd[Q]{b}{x}$ for each $x\in Q$, i.e., $f_b$ is the function that maps each point to its $w$-distance to $b$ in $Q$. 
Note that the refined triangles of $b$ in the refined FVD of $\vd[Q]{B}$ provide a representation of this $w$-distance inside $\cell[Q]{b, B}$.
We say that the geodesic path $\p[Q]{b}{x}$ is the \emph{witness path} of the value of $f_b(x)$.

\begin{observation}\label{obs: f_s and distance coincide in cell}
Let $b\in B$.
Given $x\in P$, it holds that $f_b(x) \leq \dd{b}{x}$. Moreover, if $x\in \cell{b, \s}$, then $f_b(x) = \dd{b}{x}$. 
\end{observation}

Let $r\in R$ and recall that $G_r$ is the polygon associated with $r$ defined in Section~\ref{section:First phase}. 
As a preprocessing, we have computed the SPM of $r$ within $G_r$.
We define a function $f_r:G_r\to \mathbb{R}$ that encodes the $w$-distances from $r$ with respect to the polygon $G_r$, instead of $Q$. 
That is, $f_r(x) = \dd[G_r]{x}{r}$ for each $x\in G_r$.
In this case we say that $\p[G_r]{r}{x}$ is a \emph{witness path} of the value of $f_r(x)$. 
Note that by Lemma~\ref{lemma:SPMForSitesR}, the functions $f_r$ jointly cover polygon $Q$. 

\begin{lemma}\label{lemma: f_r and distance coincide in cell}
Let $r\in R$.
Given $x\in P$, it holds that $f_r(x) \leq \dd{r}{x}$. Moreover, if $x\in \cell{r, \s}$ and the path $\p{r}{x}$ contains no point of $\bcell{r, \s}$ other than its endpoints, then $f_r(x) = \dd{r}{x}$. 
\end{lemma}
\begin{proof}
By Lemma~\ref{lemma:Voronoi cell in funnel}, we know that if $x\in \cell{r,\s}$, then  $x$ belongs to $\funnel{r, \s}$.
Moreover, since $\funnel{r, \s}$ is contained in $G_r$, we know that 
\[ f_r(x) = \dd[G_r]{r}{x} \leq \dd[\funnel{r, \s}]{r}{x} = \dd{r}{x}.\]

If $\p{r}{x}$ contains no point of $\bcell{r, \s}$, then we claim that $\p{r}{x} = \p[G_r]{r}{x}$. 
If this claim is true, then clearly $f_r(x) = \dd[G_r]{r}{x} = \dd{r}{x}$ proving our result.
To prove our claim, notice that path $\p{r}{x}$ is contained in $\funnel{r, \s}$. Moreover, because this path contains no point of $\bcell{r, \s}$ other than its endpoints, then all reflex vertices along it must belong to the walls of $\funnel{r, \s}$. However, all these reflex vertices are also part of $G_r$. Thus, since $\funnel{r, \s}\subseteq G_r$, we conclude that $\p{r}{x} = \p[G_r]{r}{x}$ proving our claim.  
\end{proof}

%To avoid degenerate situations, we assume that for no vertex of $v$ of $Q$ it holds that $f_s(v) = f_t(v)$ for two different sites $s,t\in \s$. 
%This can be obtained by applying a slight perturbation from either the weights of the sites, or their positions~\cite{edelsbrunner1990simulation}.

Note that for each site of~$R$, we have considered their $w$-distances inside of $G_r$, while for the sites in $B$, their $w$-distances are with respect to $Q$.
Therefore, in the intermediate steps of our incremental construction, we will not have the FVD of the sites, but some Voronoi-like structure.

\subsection{The envelope}
Note that $\vd[Q]{B}$ represents already the upper envelope of the functions $f_b$ for the sites in~$B$. 
We would like to complete this envelope by incrementally inserting the functions $f_r$ for the sites in $R$.
To deal with these upper envelopes, we introduce some definitions.

Consider the order of the sites of $R$ according to the random permutation $\Pi$ used to construct $B$ and $R$.
Let $\s_0 = B$ and for each $1\leq i\leq |R|$, let $\s_i$ be the set consisting of $B$ and the first $i$ red sites according to the permutation $\Pi$.
This is the order that we use for our randomized incremental construction. 
That is, on each insertion step we would like to maintain a Voronoi-like structure for the sites in $\s_i$.

Let $s$ be a site of $\s_i$.
Given a point $x$ of $Q$, we say that $x$ is \emph{\idom} by $s$ if $f_s(x) \geq f_{s'}(x)$ for all $s'\in \s_i$.
Notice that if a point $x$ is \idom, then the witness path of $f_s(x)$ must be defined. 

\begin{lemma}\label{lemma:Shadow points}
Let $s$ be a site of $\s$.
Given a point $x$ such that $x$ is \idom by $s$, and a point $z\in Q$ such that $x$ lies on the witness path of $f_s(z)$,
it holds that $z$ is \idom by $s$. 
\end{lemma}
\begin{proof}
Let $\gamma$ be the witness path of $f_s(z)$, and let $\gamma(x, z)$ denote the subpath of $\gamma$ connecting $x$ with $z$.
We prove the result by induction on the number of vertices of $\gamma(x, z)$.
For the base case, if $\gamma(x, z)$ consists of the single vertex $x = z$, then the results hold trivially. 
Let $x , v_1, \ldots, v_t,  z$ be the sequence of vertices of $\gamma(x, z)$. 
By induction hypothesis, we know that $v_t$ is \idom by $s$.
It remains only to show that the last segment $v_t z$ of this path is also \idom by $s$.
To show this, let $y$ be a point in the segment $v_t z$.
To prove that $y$  is also \idom by $s$, consider any site $s'\in \s_i$. 
Because each $f_s$ is a $w$-distance function measured from $s$ in some polygon, we know that the triangle inequality holds, and hence $f_{s'}(v_t) + |v_t y| \geq f_{s'}(y)$.
Because $v_t$ is \idom by $s$, we know that $f_s(v_t) \geq f_{s'}(v_t)$ and hence, we conclude that 
$f_s(y) = f_s(v_t) + |v_t y| \geq f_{s'}(v_t) + |v_t y| \geq f_{s'}(y)$.
That is, any point $y\in v_t z$ is \idom by $s$ as claimed.
Therefore, by induction the entire path $\gamma(x, z)$ is \idom by~$s$.
\end{proof}

Let $x$ be a point that is \idom by $s$. Extend the last segment of the witness path of $f_s(x)$ until it touches the boundary of $Q$ at a point $x^*$.
We say that $x^*$ is the \emph{$s$-shadow} of $x$. A direct consequence of Lemma~\ref{lemma:Shadow points} is the following result.

\begin{corollary}\label{corollary: Shadows in cell as well}
Let $s \in \s_i$.
If $x$ is \idom by $s$, then its $s$-shadow is also \idom by $s$ and lies on $\partial Q$.
\end{corollary}

The following result is crucial to guarantee the the resulting structure after the incremental construction coincides with the desired FVD of $\s$.

\begin{lemma}\label{lemma:Patch contains vcell} 
For each $0\leq i \leq |R|$ and for each site $s\in \s_i$, 
each point in the Voronoi cell $\cell{s, \s}$ is \idom by $s$.
\end{lemma}
\begin{proof}
Let $b\in B$ and let $x\in \cell{b, \s}$. 
Because $f_b(x) = \dd{b}{x}$ by Observation~\ref{obs: f_s and distance coincide in cell}, and since $\dd{b}{x} \geq \dd{s}{x}$ for any $s\in \s_i$ by definition of Voronoi cell, we know that $f_b(x) \geq \dd{s}{x}$ for any $s\in \s_i$. Moreover, as $x\in P$, Observation~\ref{obs: f_s and distance coincide in cell} and Lemma~\ref{lemma: f_r and distance coincide in cell} imply that $\dd{s}{x} \geq f_s(x)$. 
Thus, we conclude that $f_b(x) \geq f_s(x)$ for any $s\in \s_i$, i.e., $x$ is \idom by $s$. That is, $\cell{b, \s}$ is \idom by $b$ for each $b\in B$.

We prove now an analogous result for the red sites.
Let $r\in R$, and let $x\in \cell{r, \s}$.  We have two cases, if the path $\p{r}{x}$ contains no point of $\bcell{r, \s}$, then by Lemma~\ref{lemma: f_r and distance coincide in cell} we know that $f_r(x) = \dd{r}{x}$ and the same argument as for the blue sites described above applies. That is, in this case $x$ is \idom by $r$.
For the other case, assume that the path $\p{r}{x}$ contains a point of $\bcell{r, \s}$, and let $y$ be the first point of $\bcell{r, \s}$ in this path when going from $r$ to $x$.
Because the path $\p{r}{y}$ contains no interior point of $\bcell{r, \s}$, Lemma~\ref{lemma: f_r and distance coincide in cell} applies, and hence we conclude that $y$ \idom by $r$.
Because $y$ lies in the witness path $\p[G_r]{r}{x} = \p{r}{x}$ of $f_r(x)$, and since $y$ is \idom by $r$, Lemma~\ref{lemma:Shadow points} implies that $x$ \idom by $r$.
Therefore, regardless of the case, we know that for each $s\in \s_i$, each point in $\cell{s, \s}$ is \idom by $s$.
\end{proof}

\subsection{The insertion process}\label{section: Insertion process}

Let $r$ be the $i$-th site of $R$ inserted in our randomized incremental construction. 
To simplify our incremental construction, instead of constructing the entire set of points that are \idom by $r$, which might contain several connected components,
we focus exclusively on constructing the connected component containing $\cell{r, \s}$.
This simplifies the structure of the upper envelope, and helps us to prove a bound on its complexity. 

We define the \emph{envelope-graph} of $\s_i$ recursively. 
For $i = 0$, the envelope-graph of $\s_0$ is simply the Voronoi-tree of $\vd[Q]{B}$.
This envelope-graph induces a decomposition of $Q$ into \emph{\icells}.
The \emph{\icell} of each site $s\in \s_0$ is the connected component in this decomposition that contains $\cell{s, \s}$.

Given the envelope-graph of $\s_{i-1}$, the envelope-graph of $\s_i$ is defined as follows.
We consider the set of all points of $Q$ that are $i$-dominated by $r$ and the connected components that they induce. 
The \emph{\icell} of $r$ is the connected component that contains $\cell{r, \s}$ induced by these points.
The envelope-graph of $\s_i$ is then obtained by adding to it the boundary of the \icell of $r$, and removing everything inside it. 
In this way, the \icells[(i-1)] of the envelope graph of $\s_{i-1}$ might shrink. 
However, Lemma~\ref{lemma:Patch contains vcell} guarantees that for each site $s\in \s_i$, the Voronoi cell $\cell{s, \s}$ is $i$-dominated by $s$.
Therefore, $\cell{s, \s}$ is still contained in the \icell of $s$, i.e., the \icell of $s$ is non-empty.

\begin{lemma}\label{lemma:Structure of the envelope}
The envelope-graph of $\s_i$ is a tree with at most $2|\s_i|$ leaves lying on the boundary of $Q$.
\end{lemma}
\begin{proof}
We prove the result by induction $i$. For the base case the result holds for the Voronoi tree of $\vd[Q]{B}$ which has $|B| = |\s_0|$ leaves on the boundary of $Q$.
When inserting the $i$-th site $r$ of $R$ into the envelope-graph of $\s_{i-1}$, we know that the \icell of $r$ is connected by construction.
We claim further that the \icell of $r$ intersects $\partial Q$ in a single connected component.
Notice that this \icell intersects $\partial Q$ by Corollary~\ref{corollary: Shadows in cell as well}.
To show that this intersections consists of a single component, we assume for a contradiction that the \icell of $r$ intersects $\partial Q$ in two or more connected components. 
Because the entire \icell of $r$ is contained in $G_r$ by definition of $f_r$,  by removing the \icell of $r$ from $G_r$, we would obtain several connected components, one of them, say $K$, contained in $G_r$. However, in this case we know that for each $x\in K$, the witness path of $f_r(x)$ is defined, and crosses the \icell of $r$; see Figure~\ref{fig:EnvelopeGraph}. 
That is, this witness path contains a point, say $y$, that is \idom by $r$. 
Thus, Lemma~\ref{lemma:Shadow points} implies that $x$ is \idom by $r$. 
Because this also holds for any point on the portion of the witness path of $f_r(x)$ connecting $y$ with $x$, 
we conclude that $x$ and $y$ are connected by points that are \idom by $r$. 
Therefore, $x$ must belong to the \icell of $r$ leading to a contradiction. 
Thus, we proved that the \icell of $r$ intersects $\partial Q$ in a single connected component. 
In particular, this implies that its boundary consists of a single connected path with endpoints on $\partial Q$.

\begin{figure}[ht]
\centering
\includegraphics{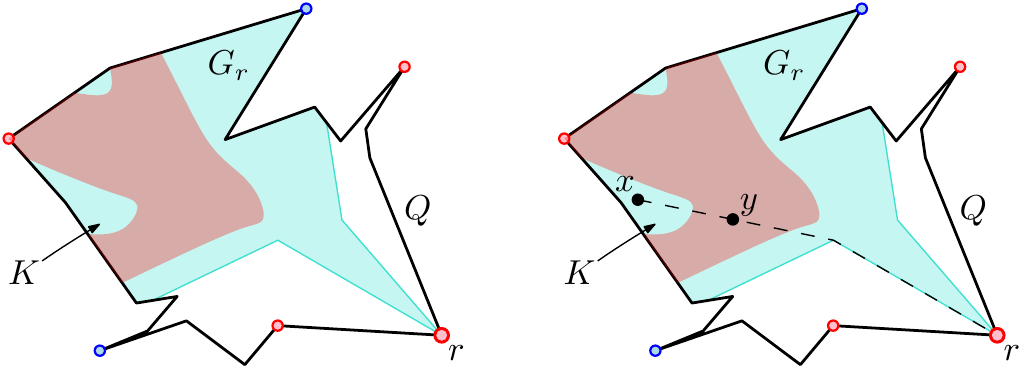}
%[width=1\textwidth]
\caption{The proof of Lemma~\ref{lemma:Structure of the envelope}. The \icell of $r$ is denoted in red and contained in $G_r$. We are assuming for a contradiction that the intersection of this \icell with $\partial Q$ is not connected.}
\label{fig:EnvelopeGraph}
\end{figure}

Note also that all \icells are non-empty by Lemma~\ref{lemma:Patch contains vcell}.
Moreover, Lemma~\ref{lemma:Shadow points} guarantees that the envelope-graph contains no bounded cycles. 
Therefore, each \icell of a site of $\s_i$ intersects the boundary of $Q$. 
Because the envelope-graph of $\s_{i-1}$ is a tree, by adding the path bounding the \icell of $r$, and removing all that this \icell encloses from the graph,
we obtain a new acyclic graph with exactly one new \icell. That is, the envelope-graph of $\s_i$ is a tree.

Because the insertion of $r$ added at most 2 new leaves to the envelope-graph, by the induction hypothesis we get that the envelope-graph of $\s_i$ has at most $2|\s_{i-1}| + 2 = 2|\s_i|$ leaves concluding our proof.
\end{proof}

\begin{lemma}\label{lemma: Patch in G_r}
The \icell of $r$ is contained in $G_r$ and intersects no point of the walls of $G_r$.
\end{lemma}
\begin{proof}
Because $f_r$ is only defined in $G_r$, we get trivially that the \icell of $r$ is contained in $G_r$. To prove the second part of the result,
recall that the walls of $G_r$ are defined by the paths $\p{r}{u_r}$ and $\p{r}{v_r}$, where $u_r$ and $v_r$ lie strictly inside of the Voronoi cells of some blue sites (see Section~\ref{sec: Preprocessing of red sites}).
By Lemma~\ref{lemma:Patch contains vcell}, we know then that $u_r$ and $v_r$ lie strictly inside the \icells of some blue sites in $B$.
Therefore, if there was a point of the \icell of $r$ on any of the walls of $G_r$, say on $\p{r}{u_r}$, then by Lemma~\ref{lemma:Shadow points} $u_r$ would belong to the \icell of $r$---a contradiction as $u_r$ lies strictly inside the \icell of a blue site. 
\end{proof}

\subsection{Algorithmic description} 

We proceed now to describe algorithmically how to carry on the incremental construction described above, and construct the envelope-graph of $\s_i$.
Our algorithm starts with the refined FVD of $\vd[Q]{B}$, and on each round constructs the boundary to the \icell of a new site of $R$. 
In addition to our envelope-graph, we maintain a set of refined triangles that cover each \icell in the same way that they cover the Voronoi cells in the refined FVD; see Figure~\ref{fig:RrefinedDiagram}.
We call this representation the \emph{refined envelope} of $\s_i$.
We assume inductively that the envelope-graph of $\s_i$ $\vd[Q]{B}$ is stored as a refined envelope. 
For the base case this holds as we assume that we have at hand the refined FVD of $\vd[Q]{B}$.
In addition, we assume that for each vertex $v$ of $Q$, we know the site whose \icell contains $v$. 
Moreover, we assume also that we have a pointer to the refined triangle of this \icell that contains $v$.

\begin{figure}[ht]
\centering
\includegraphics{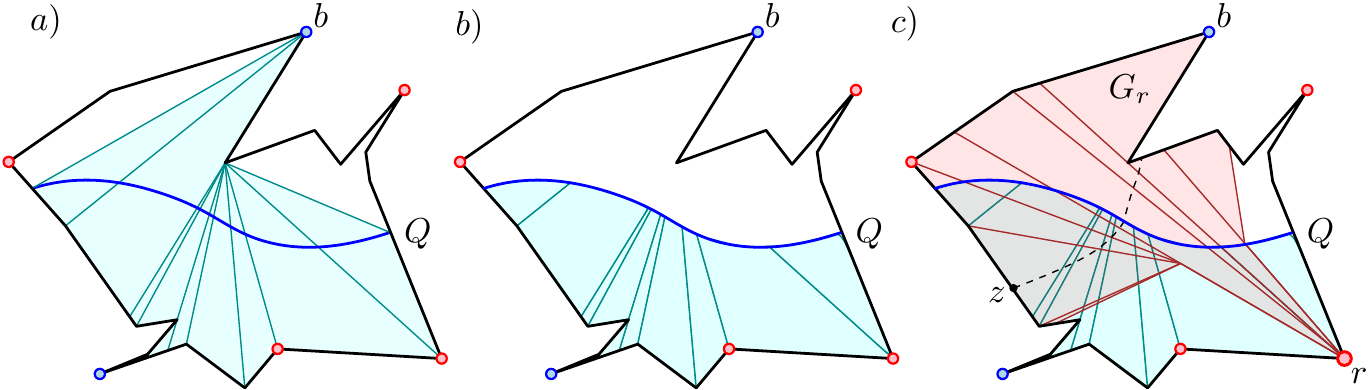}
%[width=1\textwidth]
\caption{$a)$ a site $b\in B$ and the defining triangles of $f_b$. 
$b)$ The refined triangles obtained by intersecting with the \icell[0] of $b$.
$c)$ The insertion of the red site $r$ and the update of the envelope. }
\label{fig:RrefinedDiagram}
\end{figure}

For each $b\in B$, let $\mu_b$ be the set of refined triangles that belong to $b$. 
For a site $r\in R$, let $\mu_r$ denote the set of triangles used to describe $f_r$, i.e,  the set of triangles in the SPM of $r$ inside of $G_r$.
Thus regardless of the case $\mu_s$ denotes a set of triangles (and their associated distance function) \emph{owned} by $s$.

To compute the \icell of $r$, the first step is to find a point lying on its boundary.
Note that by Lemma~\ref{lemma:Patch contains vcell}, we can easily find some points on $\partial Q$ that lie in the \icell of $r$.
Namely, we can take an endpoint $\ell$ of $\bcell{r, \s}$. 
Because $\ell\in \LL$, we know that it is a vertex of~$Q$.
Thus, we know the triangle $\triangle$ of the refined envelope of $\s_{i-1}$ that contains $\ell$, and the site $s\in \s_{i-1}$ that owns $\triangle$.
Moreover, we also know the triangle $\triangle'$ of $\mu_r$ that contains~$\ell$ in the SPM of~$r$.
Let $H$ be the intersection of $\triangle$ and $\triangle'$.
Because both $f_r$ and $f_s$ have constant description inside of $H$, we can compute their upper envelope in this small domain.
If there is a point $x$ in $H\cap \partial Q$ such that $f_r(x) = f_s(x)$, then we have found an endpoint of the boundary of the \icell of $r$.
Otherwise, we know that $H\cap \partial Q$ is entirely contained in the \icell of $r$, in this case, we can move counterclockwise to a neighboring triangle.
Note that the counterclockwise endpoint of $H\cap \partial Q$ belongs to the boundary of either $\triangle$ or $\triangle'$. 
In the former case, we move to the next refined triangle along $\partial Q$ of the refined envelope of $\s_{i-1}$, 
or in the latter case to the neighboring triangle of $\triangle'$ in $\mu_r$.
After moving, we can redefine $H$ and repeat the process until finding an endpoint $z$ of the boundary of the \icell of $r$; see Figure~\ref{fig:RrefinedDiagram}. 
Lemma~\ref{lemma: Patch in G_r} implies that $z$ will be reached before reaching a point on the wall of $G_r$.
The time needed to find $z$ is then proportional to the number of visited triangles. 

Let $D_r$ be the number of arcs of the envelope-graph of $\s_{i-1}$ that \emph{disappear} in $\s_i$, i.e., that are completely contained in the \icell of~$r$.
We claim that each refined triangle visited in the search for $z$ (except for the last one) defines an arc of  the  envelope-graph of $\s_{i-1}$ that disappears. 
If this claim is true, then the running time needed to find $z$ is at most $O(\mu_r + D_r)$.
To prove our claim, consider a refined triangle $\triangle$ of the refined envelope of $\s_{i-1}$ belonging to some site $s\in \s_{i-1}$ that is visited. 
Then there is an arc $a$ of this envelope-graph defined by this $\triangle$. 
Notice that by Lemma~\ref{lemma:Shadow points}, if arc $a$ does not disappear, then all points in $\triangle$ that lie on $\partial Q$ must belong to the \icell of $s$. 
However, we know by our construction that $f_r > f_s$ inside of $\triangle\cap \partial Q$---a contradiction. 
Thus, arc $a$ cannot be in the envelope-graph of $\s_i$ proving our claim. 

Because each visited triangle of the refined envelope of $\s_{i-1}$ (except maybe the last one) corresponds to an arc that disappears, we know that the total number of visited triangles during the search for $z$ is $O(D_r + |\mu_r|)$.
Thus, as constructing $H$ and the upper envelope inside $H$ takes $O(1)$ time, the total time to find $z$ is $O(D_r + |\mu_r|)$.

We proceed now in a similar way to trace the boundary of the \icell of $r$ inside $Q$ starting from $z$.
Note that by Lemma~\ref{lemma: Patch in G_r}, this boundary is entirely contained in the interior of $G_r$ (except for its endpoints). 
When we found~$z$, we know the structure of the boundary of the \icell of $r$ inside of $H$. 
Note that inside of $H$, this boundary has constant description. 
In fact, we know that it is either a hyperbolic arc, or a straight-line segment, being the locus of points where $f_r(x) = f_{s'}(x)$, where $s'$ is the site who owns the refined triangle of the refined envelope of $\s_{i-1}$ defining $H$.
Thus, we can move along this arc until it hits the boundary of $H$. 
At this point, we move either to the neighboring refined triangle in the refined envelope of $\s_{i-1}$, 
or the neighboring triangle in the description of $f_r$,  
depending on what boundary of $H$ this arc hits. 
Proceeding in this fashion, we can reconstruct the boundary of the entire \icell of $r$. 
Moreover, each time that we move to a neighboring refined triangle, we obtain a new arc of the boundary of the \icell of $r$. 
Thus, the number of times we have to move corresponds to the size of the \icell of $r$. 
Because each operation takes constant time, the total running time is linear on the \emph{size} of the \icell of $r$, i.e., the number of arcs and vertices that define its boundary.

Once we have computed the boundary of the \icell of $r$, we must simply crop each triangle in the description of $f_r$ that defines and arc bounding the \icell of $r$. 
In this way, we obtain the refined triangles of the \icell of $r$.
In a similar way, we can update the visited triangles of the refined envelope of $\s_{i-1}$ by removing their portion inside of the \icell of $r$.
Moreover, we know which arcs of the envelope-graph of $\s_{i-1}$ disappear, so we can remove them and update the graph to obtain both the envelope-graph and the refined envelope of $\s_i$.

\begin{lemma}\label{lemma:Time to insert one site of R}
Let $r$ be the $i$-th site of $R$ inserted in our randomized incremental construction. 
The \icell of $r$ can be computed in $O(M_r + D_r + |\mu_r|)$ time, where $M_r$ is the size of the \icell of $r$, and $D_r$ is the number of arcs of the envelope-graph of $\s_{i-1}$ that disappear.
Moreover, the refined envelope of $\s_i$ can be obtained within the same time from that of $\s_{i-1}$.
\end{lemma}

We say that the \emph{complexity} of the envelope-graph of $\s_i$ is the number of vertices and arcs defining it.

\begin{lemma}\label{lemma: Complexity of envelope-graph}
The expected complexity of the envelope-graph of $\s_i$ is $O(n)$.
\end{lemma}
\begin{proof}
Recall that for each $b\in B$, $\mu_b$ denotes the set of refined triangles that belong to $b$. 
Because $\vd[Q]{B}$ is a FVD of $B$ in $Q$, its Voronoi tree consists of $O(n)$ arcs and vertices.
Therefore, $\sum_{b\in B} |\mu_b| = O(n)$.
Also, for a site $r\in R$, $\mu_r$ denotes the set of triangles used to describe $f_r$ inside of $G_r$.
Notice that Lemma~\ref{lemma:SPMForSitesR} implies that  $\ex{\sum_{r\in R} |\mu_r|}= O(n)$.
Thus, because $\s = B\cup R$, we know that  $\ex{\sum_{s\in \s} |\mu_s|} = O(n)$.

Let $\Gamma_i = \bigcup_{s\in \s_i} \mu_s$ be the set containing all the triangles defined by the sites in $\s_i$.
Thus, $|\Gamma_i| = \sum_{s\in \s_i}|\mu_s|$.
Let $Z_i$ be the complexity of the envelope-graph of $\s_i$. 
We claim that $Z_i= O(|\Gamma_i|)$. If this claim is true, then 
\[ \ex{Z_i} = O(\ex{|\Gamma_i|})= \ex{\sum_{s\in \s_i} |\mu_s|} = O(n). \]

Thus it only remains to prove indeed that $Z_i= O(|\Gamma_i|)$.
Because the envelope-graph is a plane graph and has $|\s_i|$ leaves by Lemma~\ref{lemma:Structure of the envelope}, it has at most $|\s_i| - 2$ vertices of degree larger than two.
Thus, it remains only to account for the vertices of degree two.
From the algorithmic construction described above,
we know that these vertices of degree two are created when the algorithms reaches the boundary a triangle in $\Gamma_i$ while tracing the boundary of the new \icell. 
We claim each straight line edge bounding a triangle of $\Gamma_i$ contains at most one vertex of degree two of the envelope-graph of $\s_i$.
From this claim, and since $|\s_i|\leq |\Gamma_i|$, it follows that  $Z_i$ consists of $O(|\Gamma_i|)$ vertices completing our~proof.
Thus, it only remains to prove this last claim.

Before jumping into the proof, we need to observe the following.
Consider an arc $a$ of the envelope-graph bounding the \icell of some site $s\in \s_i$.
If $x$ is a point in the interior of~$a$, then by Corollary~\ref{corollary: Shadows in cell as well}, 
the open segment $x x^*$ crosses no other arc of the envelope-graph, where $x^*$ is the $s$-shadow of $x$.

To prove that each edge bounding a triangle of $\Gamma_i$ contains at most one vertex of degree two, we proceed as follows.
Assume for a contradiction that two vertices $u$ and $v$ of degree two lie on a single edge $e$ bounding a triangle $\triangle\in \Gamma_i$.
Assume that $e$ is vertical and that $\triangle$ belongs to the description of $f_C$ for some $r\in R$. 
The case where $\triangle$ is a refined triangle is analogous. Assume without loss of generality that $u$ lies in $\p{s}{v}$.
Note that each of $u$ and $v$ have two arcs of the envelope-graph incident to them.
Regardless of the case, we can assume without loss of generality that two of them go to the right of $e$. 
That is, there must be a point $x$ in the interior of one of them such that the segment $x x^*$ intersects the other arc---a contradiction with our observation above.
Therefore, each edge bounding a triangle of $\Gamma_i$ contains at most one vertex of degree two proving our claim.
\end{proof}

We are now ready to provide the proof of our main results by combining the lemmas presented in this section. 

\begin{theorem}
The envelope-graph and refined envelope of $\s_{|R|}$ can be computed in expected $O(n)$ time.
Moreover, for each $s\in \s$, the envelope-graph of $\s_{|R|}$ coincides with the Voronoi tree of $\vd{\s}$.
\end{theorem}
\begin{proof}
By Lemma~\ref{lemma:Patch contains vcell}, the \icell[|R|] of each site $s$ of $\s$ contains its corresponding Voronoi cell $\cell{s, \s}$. 
Because the union of these Voronoi cells covers $P$, we conclude that the \icell[|R|] of $s$ and $\cell{s, \s}$ coincide inside $P$ for each $s\in \s$.
That is, the envelope-graph of $\s_{|R|}$ coincides with the Voronoi tree of $\vd{\s}$.

Let $r$ be the $i$-th site of $R$ inserted in our randomized incremental construction. 
Notice that by Lemma~\ref{lemma:Time to insert one site of R}, 
the time needed to insert $r$ is $O(M_r + D_r + |\mu_r|)$, where $M_r$ is the size of the \icell of $r$, $D_r$ is the number of arcs of the envelope-graph of $\s_{i-1}$ that disappear, and $\mu_r$ is the set of triangles defining $f_r$. 
Thus, the total running time of our incremental construction is 
\[ O\left(\sum_{r\in R} (M_r + D_r + |\mu_r|) \right ) = O\left( \sum_{r\in R} M_r + \sum_{r\in R} D_r + \sum_{r\in R} |\mu_r|\right). \]

We claim that $\ex{M_r} = O(n/m)$, where $m = |\s|$.
If this claim is true, then by Lemma~\ref{obs: Complexity of B and R} and linearity of expectation, 
\[\ex{\sum_{r\in R} M_r} = O\left(\frac{|R| n}{m}\right) =  O(n).\]
Moreover, we know that $\ex{\sum_{r\in R} |\mu_r|} = O(n)$ by Lemma~\ref{lemma:SPMForSitesR} and the definition of $\mu_r$.
Finally, since an arc can only be created once, then the total number of arcs that disappear is at most the number of arcs created, which is the number of arcs in $\vd[Q]{B}$, plus the arcs accounted in each $M_r$. Thus, by our previous arguments $\ex{\sum_{r\in R} D_r}  \leq  O(n) + \ex{\sum_{r\in R} |\mu_r|} = O(n)$.
Consequently, if our claim about the expected value of $M_r$ is true, then the expected total running time to insert $|R| = O(m)$ sites, and compute the envelope-graph of $\s_{|R|}$ is 
\[O\left( \ex{\sum_{r\in R} M_r + \sum_{r\in R} D_r + \sum_{r\in R} |\mu_r|}\right) = O(n).\]

Thus, it only remains to bound the expected size of the \icell of $r$, i.e., to show that $\ex{M_r} = O(n/m)$.
That is, we can bound its expected size independent of $i$.
For a site $s\in \s_i$, let $X_s$ be an i.r.v. that is one if and only if $s = r$ such that $Pr[X_s = 1] = \frac{1}{|\s_i|}$.
Additionally, let $M_s$ be the \emph{size} of the \icell of $s$, i.e., the number of arcs and vertices that define its boundary. 
Using this notation,
\[M_r = \sum_{s\in \s_i} ( X_s M_s), \text{ and by taking expectation, } \ex{M_r} = \sum_{s\in \s_i} \ex{X_s M_s}.\]
Note that $X_s$ and $M_s$ are independent by definition of $X_s$.
Thus, we get that  
\[\ex{M_r} = \sum_{s\in \s_i} \ex{X_s} \ex{M_s} = \frac{1}{|\s_i|} \ex{\sum_{s\in \s_i} M_s}.\]
Note that $\sum_{s\in \s_i} M_s$ coincides with the complexity of the envelope-graph of $\s_i$.
Thus, because Lemma~\ref{lemma: Complexity of envelope-graph} implies that $\ex{\sum_{s\in \s_i} M_s} = O(n)$, we get that  $\ex{M_r} = O(n/|\s_i|)$.
Because $\alpha m \leq |B| \leq |\s_i| \leq m$ by Lemma~\ref{obs: Complexity of B and R}, we conclude that $\ex{M_r} = O(n/m)$ as claimed.
\end{proof}

Putting everything together, we obtain the following result.

\begin{theorem}
Let  $P$ be a simple polygon and let let $\s$ be a set of $m\geq 3$ weighted sites contained in $V(P)$.
We can compute the FVD of $\s$ in $P$ in $O(n)$ time. 
\end{theorem}

{\small
\bibliographystyle{abbrvnat}
\bibliography{Geodesic}}

%\newpage 
%\appendix
%\section{Missing proofs.}
%\setcounter{theorem}{5}
%\ProofApexFVD
%\setcounter{theorem}{9}
%\ProofPropertiesOfQ
%\setcounter{theorem}{11}
%\ProofShadowPoints
%\ProofSeparatingCurveShadow
%\setcounter{theorem}{14}
%\ProofBisectorFunctions
%\setcounter{theorem}{16}
%\ProofSpmSupersetOfCell
%\ProofTimeInsertRedSite
%\ProofComplexityEnvelope

\end{document}